\documentclass{cccg24}
\usepackage{graphicx,amssymb,amsmath}

\usepackage[english]{babel}
\usepackage[colorlinks=true, allcolors=blue]{hyperref}
\usepackage{mathtools}
\usepackage[normalem]{ulem}
\usepackage{xspace}
\usepackage[linesnumbered,ruled,vlined]{algorithm2e} 
\usepackage{wasysym} 
\usepackage{subcaption}

\newcommand{\Oh}[1]{\mathcal{O}\!\lrParents{#1}}
\newcommand{\lrBrackets}[1]{\left[ #1 \right]}
\newcommand{\lrCurlyBrackets}[1]{\left\{ #1 \right\}}
\newcommand{\OhSmall}[1]{
  \mathchoice
    {{\scriptstyle\mathcal{O}}}
    {{\scriptstyle\mathcal{O}}}
    {{\scriptscriptstyle\mathcal{O}}}
    {\scalebox{.6}{$\scriptscriptstyle\mathcal{O}$}}\!\lrParents{#1}
  }
\newcommand{\OhMega}[1]{\Omega\!\lrParents{#1}}
\newcommand{\OhTheta}[1]{\Theta\!\lrParents{#1}}
\newcommand{\lrParents}[1]{\left( #1 \right)}
\newcommand{\lrSize}[1]{\left| #1 \right|}
\newcommand{\dist}[2]{d\!\left( #1, #2 \right)}
\newcommand{\f}[1]{f\!\left( #1 \right)}

\newcommand{\logn}{\log{n}}

\newcommand{\R}{\mathbf{R}}

\newcommand{\Tri}[3]{\triangle #1#2#3}
\newcommand{\TriS}[1]{\triangle #1}
\newcommand{\sens}[1]{\sigma\!\lrParents{#1}}
\newcommand{\sensO}[1]{\sigma_O\!\lrParents{#1}}
\newcommand{\hp}[1]{\overrightarrow{#1}}

\newcommand{\poly}[1]{\pentagon\!\lrParents{#1}}

\newtheorem{definition}{Definition}[section]
\newtheorem{case}{Case}
\newtheorem{innercustomthm}{Lemma}
\newenvironment{customlemma}[1]
  {\renewcommand\theinnercustomthm{#1}\innercustomthm}
  {\endinnercustomthm}

\title{Fast Area-Weighted Peeling of Convex Hulls for Outlier Detection\footnote{This work is supported in part by Independent Research Fund Denmark grant~9131-00113B and a fellowship from the Department of Computer Science at UC Irvine.}}

\author{Vinesh Sridhar\thanks{University of California, Irvine, {\tt vineshs1@uci.edu}}
    \and
    Rolf Svenning\thanks{The Department of Computer Science, Aarhus University, {\tt rolfsvenning@cs.au.dk}}}

\index{Sridhar, Vinesh}
\index{Svenning, Rolf}

\begin{document}
\thispagestyle{empty}
\maketitle

\begin{abstract}
We present a novel 2D convex hull peeling algorithm for outlier detection, which repeatedly removes the point on the hull that decreases the hull's area the most. To find $k$ outliers among $n$ points, one simply peels $k$ points.
The algorithm is an efficient \emph{heuristic} for \emph{exact} methods, which find the $k$ points whose removal together results in the smallest convex hull.
Our algorithm runs in $\Oh{n\logn}$ time using $\Oh{n}$ space for any choice of $k$. 
This is a significant speedup compared to the fastest exact algorithms, which run in $\Oh{n^2\logn + (n - k)^3}$ time using $\Oh{n\logn + (n-k)^3}$ space by Eppstein et al.~\cite{1992_Eppstein_k3n2, 1992_eppstein_kn3}, and $\Oh{n\logn + \binom{4k}{2k}(3k)^k n}$ time by Atanassov et al.~\cite{2009_Atanassov_exact_small_k}. 
Existing heuristic peeling approaches are not area-based. Instead, an approach by Harsh et al.~\cite{2016_heuristic_peeling_distance} repeatedly removes the point furthest from the mean using various distance metrics and runs in $\Oh{n\logn + kn}$ time.
Other approaches greedily peel one convex layer at a time~\cite{hugg2006_layer_depth, aloupis2006_layer_depth, peeling_layers1, peeling_layers2}, which is efficient when using an $\Oh{n\logn}$ time algorithm by Chazelle~\cite{1985_chazelle_convex_layers} to compute the convex layers. However, in many cases this fails to recover outliers.
For most values of $n$ and $k$, our approach is the fastest and first practical choice for finding outliers based on minimizing the area of the convex hull. 
Our algorithm also generalizes to other objectives such as perimeter.


\end{abstract}

\section{Introduction}
When performing data analysis, a critical first step is to identify outliers in the data. This has applications in data exploration, clustering, and statistical analysis~\cite{data_exploration_zuur2010protocol, clustering_DAVE1991657, statistical_analysis_kwak2017statistical}.
Typical methods of outlier detection such as Grubbs' test~\cite {grubbs1949sample} are based in statistics and require strong assumptions about the distribution from which the sample is taken. 
These are known as parametric outlier detection tests.
If the sample size is too small or the distribution assumptions are incorrect, parametric tests can produce misleading results. 
For these reasons, non-parametric complementary approaches based in computation geometry have emerged. 
Our work follows this line of research and is based on the fundamental notion of a convex hull. 
For a set of points $P$, the convex hull is the smallest convex set containing $P$~\cite{de2000computational}. 

There are numerous definitions of outliers~\cite{knorr1999finding, ramaswamy2000efficient, angiulli2002fast}, but a general theme is that points without many close neighbors are likely to be outliers. 
As such, these outlying points tend to have a large effect on the shape of the convex hull. Prior work has applied this insight in different ways to identify possible outliers, such as removing points from the convex hull to minimize its diameter~\cite{aggarwal1989fining, eppstein1994_exact_smallest_diameter}, its perimeter~\cite{dobkin1983finding}, or its area~\cite{1992_eppstein_kn3, 1992_Eppstein_k3n2}. 
Motivated by the last category, we will consider likely outliers to be points whose removal causes the area of the convex hull to shrink the most. 
We propose a greedy algorithm that repeatedly removes the point $p \in P$ such that the area of $P$'s convex hull decreases the most. We call the amount the area would decrease if point $p$ is removed its \emph{sensitivity} $\sens{p}$. 
The removed point is guaranteed to be on the convex hull, and such an algorithm is known as a convex hull \emph{peeling} algorithm~\cite{peeling_layers1, peeling_layers2}. 
To find $k$ outliers, we peel $k$ points. 
Our algorithm is conceptually simple, though it relies on the black-box use of a dynamic (or deletion-only) convex hull data structure~\cite{1992_hershberger_semi_dynamic_CH, 2002_brodal_rico_dynamic_CH}. 
We assume that points are in general position. This assumption may be lifted using perturbation methods~\cite{mehlhorn2006reliable}.

\begin{figure}
    \centering
    \includegraphics[scale=0.23]{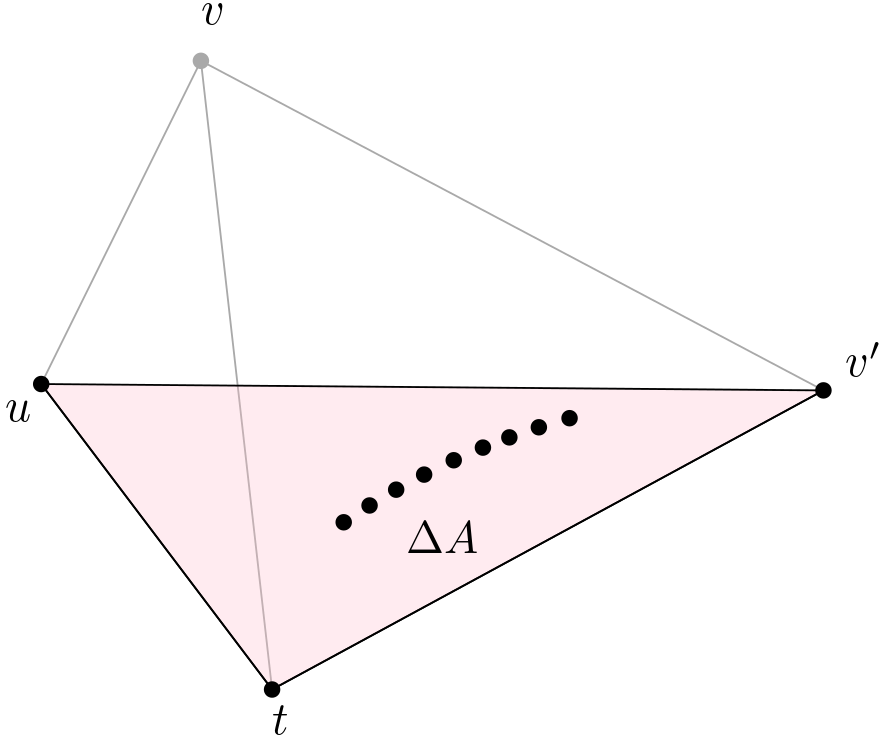}
    \caption{Here point $v$ was peeled from the convex hull and replaced by $v'$. The previous triangle $\Tri{t}{u}{v}$ for $u$ contained no points. 
    However, when $u$'s triangle becomes $\Tri{t}{u}{v'}$, the set of points $\Delta A$ affect the sensitivity $\sens{u}$ of $u$. The size of $\Delta A$ may be $\OhMega n$.
    }
    \label{fig:many-points}
\end{figure}

The main challenge is maintaining the sensitivities as points are peeled. When peeling a single point $v$, there may be $\OhMega{n}$ new points affecting the sensitivity $\sens{u}$ for a different point $u \neq v$, as in Figure~\ref{fig:many-points}. In that case, naively computing the new sensitivity $\sens{u}$ would take $\OhMega{n}$ time. Nevertheless, we show that our algorithm runs in $\Oh{n \logn}$ time for any $1 \leq k \leq n$.



\section{Related work}


The two existing approaches for finding outliers based on the area of the convex hull took a more ideal approach. 
They considered finding the $k$ points (outliers) whose removal together causes the area of the convex hull to decrease the most. 
We call this a \emph{$k$-peel} and note that it always yields an area smaller or equal to that of performing $k$ individual $1$-peels. It is not hard to come up with examples where the difference in area between the two approaches is arbitrarily large such as in Figure~\ref{fig:our-bad-example}. 
Still, these examples are quite artificial and require that outliers have at least one other point close by. 
More importantly, these methods are combinatorial in nature, and much less efficient than our algorithm. 
The state-of-the-art algorithms for performing a \emph{$k$-peel} run in $\Oh{n^2\logn + (n - k)^3}$ time and $\Oh{n\logn + (n-k)^3}$ space by Eppstein~\cite{1992_Eppstein_k3n2, 1992_eppstein_kn3} and $\Oh{n\logn + \binom{4k}{2k}(3k)^k n}$ time by Atanassov et al.~\cite{2009_Atanassov_exact_small_k}. While excellent theoretical results, for most values of $1 \leq k \leq n$ and $n$, the running time of both of these algorithms is prohibitive for practical purposes. Our contribution is a fast and practical heuristic for these ideal approaches. There are also several results for finding the $k$ points minimizing other objectives such as the minimum diameter,   perimeter, or area-enclosing rectangle~\cite{eppstein1994_exact_smallest_diameter, segal1998_enclosing_rectangle_exact}. 

\begin{figure}
    \centering
    \includegraphics[scale=0.1]{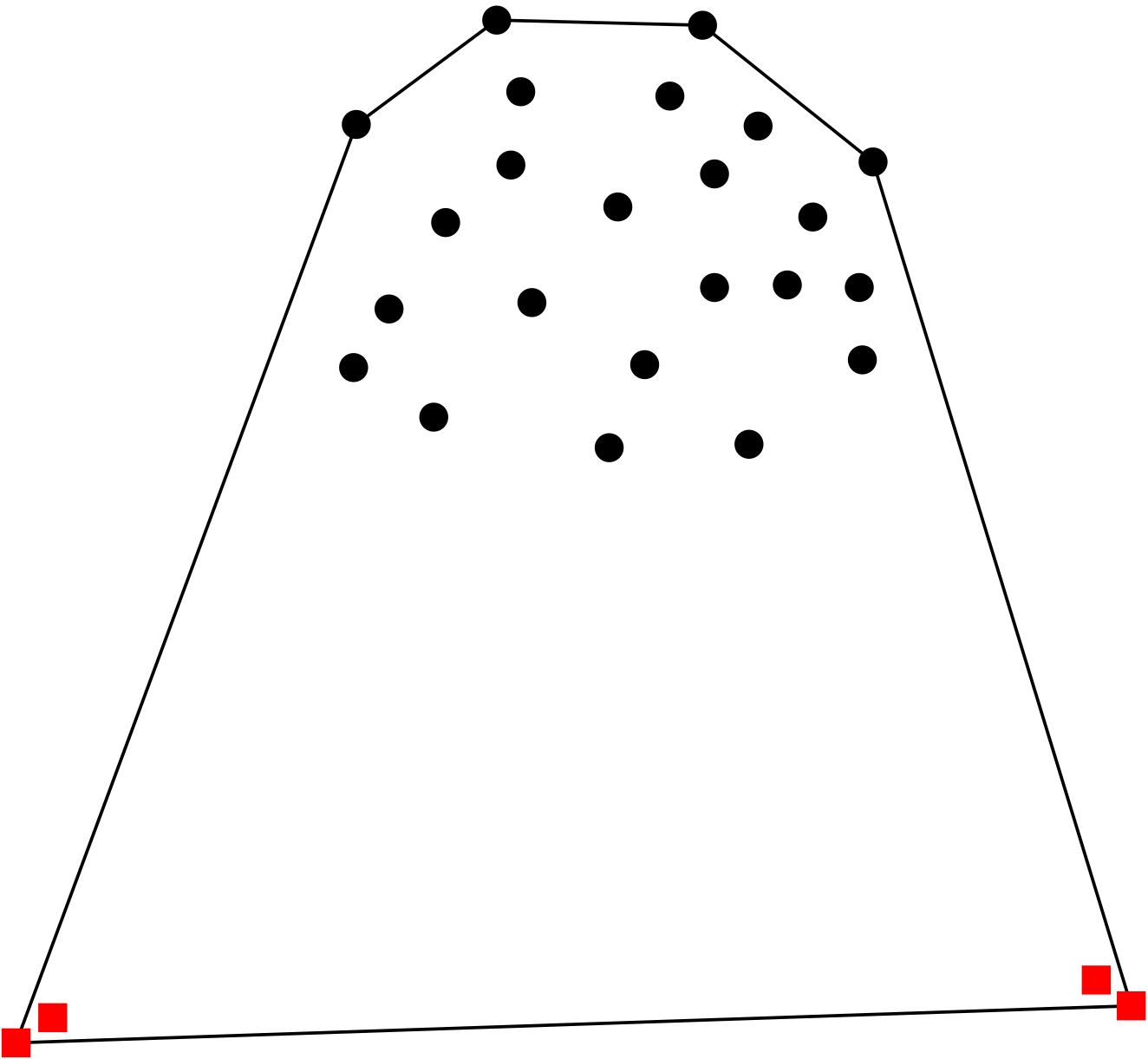}
    \caption{This figure demonstrates the limitations of our heuristic weighted-peeling approach. Clearly, the red squares are outliers, but because there are two squares close-by, the sensitivity of the red squares is minimal. Thus, our algorithm may peel all the valid points before peeling the outlier squares. Note that two $k$-peels for $k=2$ would be sufficient to remove all outliers.}
    \label{fig:our-bad-example}
\end{figure}

Another convex hull peeling algorithm is presented in~\cite{2016_heuristic_peeling_distance}. Unlike in area-based peeling, they repeatedly remove the point furthest from the mean under various distance metrics. Letting $d$ be the time to compute the distance between two points, their algorithm runs in $\Oh{n\logn + knd}$ time, which is also significantly slower than our algorithm for most values of $k$. Since they maintain the mean of the remaining points during the peeling process, each peel takes $\OhTheta{n}$ time. 

Some depth-based outlier detection methods also use convex hulls. They compute a point set's convex layers, which can be defined by iteratively computing $P \setminus CH(P)$ and are computable in $\Oh{n\logn}$ time~\cite{1985_chazelle_convex_layers}. Here, points are deleted from the outermost-layer-in~\cite{hugg2006_layer_depth, aloupis2006_layer_depth, peeling_layers1, peeling_layers2}. While efficient, the natural example in Figure~\ref{fig:motivating-example} is a bad instance for this approach.

\begin{figure}
  \centering
  \includegraphics[scale=0.1]{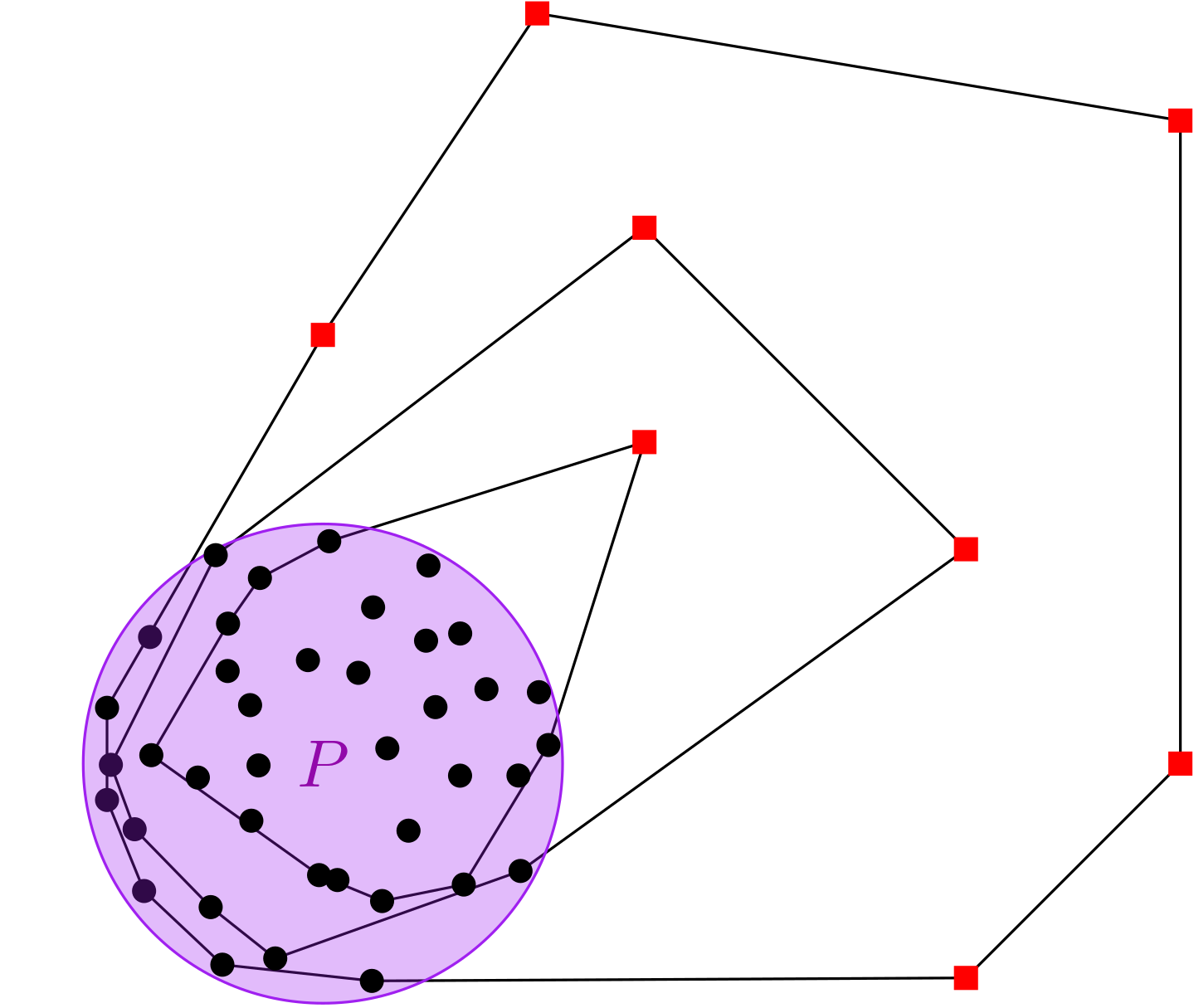}

  \caption{
    This example shows points drawn uniformly from a target disk $P$. 
    Clearly, the outliers are the points marked as red squares. 
    It shows the downside of peeling based on depth since many points have to be peeled before reaching the outliers on the second and third layers.
    In particular, if there are $n$ points drawn uniformly from $P$, then its convex hull has expected size $\Oh{n^{1/3}}$~\cite{har2011expected}. 
  }
  \label{fig:motivating-example}
\end{figure}

\section{Results}
The main result of our paper is Theorem~\ref{thm:main_peeling}, that there exists an algorithm for efficiently performing area-weighted-peeling.

\begin{theorem}\label{thm:main_peeling}
    Given $n$ points in 2D, Algorithm~\ref{alg:area_weighted_peeling} performs area-weighted-peeling, repeatedly removing the point from the convex hull which causes its area to decrease the most, in $\Oh{n\logn}$ time. 
\end{theorem}


To prove Theorem~\ref{thm:main_peeling}, we derive Theorem~\ref{thm:bounding_active_points_amortized}, which bounds the total number of times points become \emph{active} in any 2D convex hull peeling process to $\Oh{n}$. 


\begin{definition}[Active Points]
    Let $(t, u, v)$ be consecutive points on the first layer in clockwise order. A point $p$ is active for $u$ if, upon deleting $u$ and restoring the first and second layers, $p$ moves to the first layer.
\end{definition}

Intuitively, the active points are the points not on the convex hull that affect the sensitivities. 
Note that the active points form a subset of the points on the second convex layer. We define $A(u)$ to be the set of active points for point $u$ in a given configuration. Furthermore, all points in $A(u)$ can be found by performing gift-wrapping starting from $u$'s counterclockwise neighbor $t$ while ignoring $u$. We use this ordering for the points in $A(u)$. In Theorem~\ref{thm:generalization}, we show that our algorithm generalizes to other objectives such as \emph{perimeter} where the sensitivity only depends on the points on the first layer and the active points.


\section{Machinery}

In this section, we describe some of the existing techniques we use. To efficiently calculate how much the hull shrinks when a point is peeled, we perform tangent queries from the neighbours of the peeled point to the second convex layer. 
The tangents from a point $q$ to a convex polygon $L$ can be found in $\Oh{\logn}$ time both with~\cite{1981_overmars_maintenance_2D_convex_hull_tangent_with_seprating_line} and without~\cite{1995_kirkpatrick_Snoeyink_tangents_without_separating_line} a line separating $q$ and $L$. In our application, such a separating line is always available, and either approach can be used. 
Tangent queries require that $L$ is represented as an array or a balanced binary search tree of its vertices ordered (cyclically) as they appear on the perimeter of $L$. To allow efficient updates to $L$ we use a binary tree representation that is leaf-linked such that given a pointer to a vertex its successor/predecessor can be found in $\Oh{1}$ time. 

The convex layers of $n$ points can be computed in $\Oh{n\logn}$ time using an algorithm by Chazelle~\cite{1985_chazelle_convex_layers}. 
Given $l$ convex layers, after a single peel they can be restored in $\Oh{l\logn}$ time (Lemma 3.3~\cite{loffler2014unions}). 
However, for our purposes we only need the 2 outermost layers for area calculations. As such, we explicitly maintain the two outermost layers $L^1$ and $L^2$, and we store all remaining points $P \setminus \lrCurlyBrackets{L^1 \cup L^2}$ in a \emph{center} convex hull. To restore $L^1$ we use tangent queries on $L^2$ as in~\cite{loffler2014unions}. To restore $L^2$ we use extreme point queries on the center convex hull which we maintain using a semi-dynamic~\cite{1992_hershberger_semi_dynamic_CH} or fully-dynamic~\cite{2002_brodal_rico_dynamic_CH} convex hull data structures supporting extreme point queries in worst case $\Oh{\logn}$ time and updates in amortized $\Oh{\logn}$ time. 

\section{Area-Weighted-Peeling Algorithm}
In this section, we describe Algorithm~\ref{alg:area_weighted_peeling} in detail and show that its running time is $\Oh{n\logn}$. 

\begin{figure}
    \centering
    \includegraphics[scale=0.193]{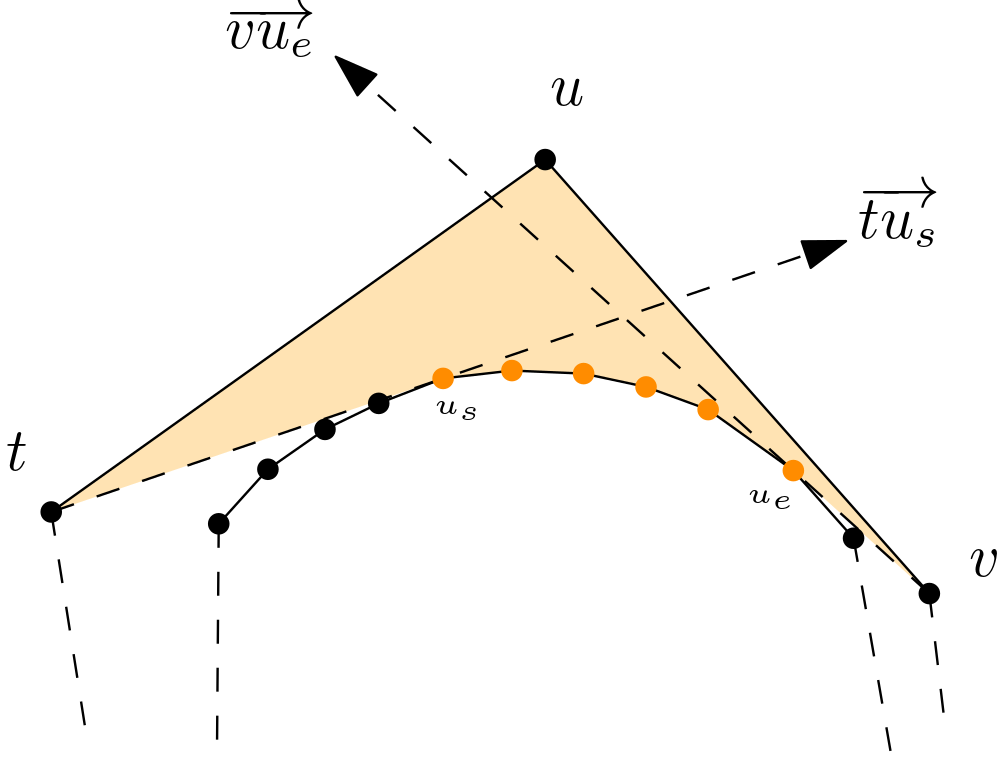}
    \caption{Using $u$'s neighbors, we can perform two tangent queries on $L^2$ to recover the first and last active point of $u$, labeled $u_s$ and $u_e$ respectively, in $\Oh{\log n}$ time. Because we represent $L^2$ as a leaf-linked tree, we can walk along $L^2$ to recover all points of $A(u)$. The shaded part of the figure represents $\sens u$.}
    \label{fig:tangent-query}
\end{figure}


At a high level, we want to repeatedly identify and remove the point which causes the area of the convex hull to decrease the most.
Such an iteration is a \emph{peel}, and we call the amount the area would decrease if point $u$ was peeled the sensitivity $\sens{u}$ of $u$. To efficiently find the point to peel, we maintain a priority queue $Q$ on the sensitivities of hull points. 
Only points on the convex hull may have positive sensitivity, and in lines~\ref{alg:init_PQ}-\ref{alg:init_sens} we compute the initial sensitives of the points on the convex hull and store them in $Q$. 
For a hull point $u$, to compute its sensitivity $\sens{u}$ we find its active points $A(u)$.
Note they must be on the second convex layer, and if $u$'s neighbors are $t$ and $v$, then the points $A(u)$ are in the triangle $\Tri{t}{u}{v}$.
In line~\ref{alg:computing_convex_layers} we compute the two outer convex hull layers represented as balanced binary trees. That allows us to compute $A(u)$ using tangent queries on the inner layer from $t$ and $v$.
Then $\sens u$ can be found by computing the area of the polygon $\poly{t \circ v \circ A(u)}$.


As points are peeled (lines~\ref{alg:peeling_start}-\ref{alg:peeling_end}) layers $L^1$ and $L^2$ must be restored. 
To restore $L^1$ when point $u$ is peeled (line~\ref{alg:peel}) we perform tangent queries on $L^2$ as in~\cite{loffler2014unions} to find $u$'s active points $A(u)$ (line~\ref{alg:find_active}) and move $A(u)$ from $L^2$ to $L^1$. See Figure~\ref{fig:tangent-query} for an example of tangent queries from $L^1$ to $L^2$.

To restore the broken part of $L^2$, we perform extreme point queries on the remaining points efficiently using a dynamic convex hull data structure $D_{CH}$ (line~\ref{alg:compute_CH_data_structure}) as in~\cite{1992_hershberger_semi_dynamic_CH} or~\cite{2002_brodal_rico_dynamic_CH}. As described in Lemma~\ref{lem:update-sensitivities}, $A(u)$ is always contiguous on $L^2$. Therefore, removing $A(u)$ from $L^2$ requires us to restore it between two ``endpoints'' $a$ and $b$. The first extreme point query uses line~$\overline{ab}$ in the direction of $u$. If a point $z$ from $D_{CH}$ is found then at least two more queries are performed with lines $\overline{za}$ and $\overline{zb}$. In general, if $k$ points are found then the number of queries is $2k + 1$. The $k$ points are deleted from $D_{CH}$. This all happens on line~\ref{alg:restore_layers_and_DCH}. 

Next, we compute the sensitivities of the new points on the hull (line~\ref{alg:compute_new_sens}) and insert them into the priority queue. 
Finally, we update the sensitivities of $u$'s neighbors $t$ and $v$ (line~\ref{alg:updating_neighbour_areas}), which, by Lemma~\ref{lem:combined}(4), are the only two points already in $Q$ whose sensitivity changes.

\begin{algorithm} 
\caption{Weighted peeling}\label{alg:area_weighted_peeling}
\KwIn{A set of $n$ points $P$ in 2D 
}

$L^1, L^2 \longleftarrow$ the first two convex layers of $P$\nllabel{alg:computing_convex_layers} \\ 
$Q \longleftarrow$ empty max priority queue \nllabel{alg:init_PQ}\\
\For{$i=1$ \KwTo $\lrSize{L^1}$}
    {   
        $u \longleftarrow L^1_i$\nllabel{alg:init_area_calc}\\
        Compute sensitivity $\sens{u}$ for $u$ \\
        $Q.insert\!\lrParents{u, \sens{u}}$ \\
        \nllabel{alg:init_sens} 
    }
    
$D_{CH} \longleftarrow$ a dynamic convex hull data structure on $P \setminus \lrCurlyBrackets{L^1 \cup L^2}$ \nllabel{alg:compute_CH_data_structure} \\

\For{$i=1$ \KwTo $n$}
    {   \nllabel{alg:peeling_start}
        $u \longleftarrow Q.extractMax$ \nllabel{alg:peel} \\
        $A(u) \longleftarrow$ $u$'s active points \nllabel{alg:find_active}\\ 
        Delete $u$ from $L^1$ and update $L^1$, $L^2$ and $D_{CH}$ \nllabel{alg:restore_layers_and_DCH}\\
        \For{$i=1$ \KwTo $\lrSize{A(u)}$}
        {   
            $\Bar{u} \longleftarrow A(u)_i$ \\
            Compute sensitivity $\sens{\Bar u}$ for $\Bar u$ \nllabel{alg:compute_new_sens} \\ 
            $Q.insert\!\lrParents{\Bar u, \sens{\Bar u}}$ \\
        }
        $t, v \longleftarrow$ neighbors of $u$ in $L^1$ \\ 
        Update $Q\!\lrBrackets{t}$ and $Q\!\lrBrackets{v}$ \nllabel{alg:updating_neighbour_areas}\\
        \nllabel{alg:peeling_end}
    }

\end{algorithm}

\subsection{Analysis}
The hardest part of the analysis is showing that the overall time spent on lines~\ref{alg:compute_new_sens}~and~\ref{alg:updating_neighbour_areas} is $\Oh{n\log n}$. 
We first show that, excluding the time spent on these lines, the running time of Algorithm~\ref{alg:area_weighted_peeling} is $\Oh{n \logn}$. 
In line 1 we compute the first and second convex layers in $\Oh{n\log{n}}$ time by running any optimal convex hull algorithm twice. 
In lines~\ref{alg:init_PQ}~to~\ref{alg:init_sens}, we compute the initial sensitivities by finding the points active for each $u \in L^1$. 
As described above, we can do this using two tangent queries on $L^2$ from the neighbors of $u$. 
Once $A(u)$ is found for each $u$, we find $\sens{u}$ by computing the area of the polygon $\poly{t \circ u \circ v \circ A(u)}$, where $t$ and $v$ are the neighbors of $u$. 
In total, we make $\Oh{|L^1|} = \Oh{n}$ tangent queries, each of which takes $\Oh{\log n}$ time. 
Since the area of a simple polygon can be computed in linear time~\cite{meister1769generalia}, and by Lemma~\ref{lem:combined}(1) each point on $L^2$ is active in at most three triangles in the initial configuration, the total time to compute all areas is $\sum_{u \in L^1} \OhTheta{1 + \lrSize{A(u)}} = \Oh{\lrSize{L^1} + \lrSize{L^2}} = \Oh{n}$ time. 
Therefore, the overall time to initialize the priority queue is $\Oh{n\log n}$. 

Initializing $D_{CH}$ in line~\ref{alg:compute_CH_data_structure} takes $\Oh{n\log n}$ time~\cite{1992_hershberger_semi_dynamic_CH}. 
In line~\ref{alg:find_active}, we can perform tangent queries on $L^2$ from $t$ and $v$ to find the first and last active points of $u$.
In line~\ref{alg:restore_layers_and_DCH}, it will take no more than $\Oh{n}$ tangent queries to restore $L^1$ and $L^2$ throughout the algorithm by charging the queries to the points moved from the center convex hull to $L^2$ or from $L^2$ to $L^1$. 
Using an efficient dynamic convex hull data structure, it takes $\Oh{\log n}$ amortized time to delete a point and thus $\Oh{n\logn}$ time overall~\cite{1992_hershberger_semi_dynamic_CH, 2002_brodal_rico_dynamic_CH}. 
We add points to the priority queue $n$ times, delete points from the priority queue $n$ times, and perform $\Oh{1}$ priority queue update operations for each iteration of the outer loop on line~\ref{alg:peeling_start}. 
Excluding lines~\ref{alg:compute_new_sens}~and~\ref{alg:updating_neighbour_areas} this establishes the overall $\Oh{n\logn}$ running time.

To bound the total time spent on line~\ref{alg:compute_new_sens} to $\Oh{n\log n}$, we prove Theorem~\ref{thm:bounding_active_points_amortized}, bounding the total number of times points becomes active to $\Oh{n}$. 
Computing $\sens{\Bar u}$ in line~\ref{alg:compute_new_sens} requires us to find $A(\Bar u)$, where $\Bar u$ is a new point added to the first layer. 
From the theorem, it takes $\Oh{n\logn}$ time to compute $A(\Bar u)$ for every $\Bar u$. In addition, because it takes $\OhTheta{1 + |A(\Bar u)|}$ to compute $\sens{\Bar u}$ from $A(\Bar u)$, overall it takes $\Oh{n}$ time to compute $\sens{\Bar u}$ for every $\Bar u$.


To bound the total time spent on line~\ref{alg:updating_neighbour_areas} on updating the sensitivities of $u$'s neighbors to $\Oh{n \log n}$, we prove Lemma~\ref{lem:update-sensitivities}. Together with Theorem~\ref{thm:bounding_active_points_amortized}, it implies the desired result.

\section{Geometric properties of peeling}
In this section, we develop an amortized analysis of peeling to show that lines~\ref{alg:compute_new_sens} and \ref{alg:updating_neighbour_areas} can be computed efficiently. We ultimately aim to show that the number of times that any point becomes active for any triangle is $\Oh{n}$, bounding the amount of work done to initialize new triangles to $\Oh{n\log n}$. 
Then we show that the amount of work done to update the sensitivities of neighbor points is proportional to the number of new active points for them and an additive $\Oh{\log n}$ term. 
Thus, updating the sensitivities over all $n$ iterations takes $\Oh{n\log n}$.

\subsection{Preliminaries}
When considering outer hull points, we use the notation $\Tri{t}{u}{v}$ for the triangle formed by $u$, its counterclockwise neighbor $t$, and its clockwise neighbor $v$. 
For a set of ordered vertices $V$ we let $\poly{V}$ be the polygon formed by the points in the (cyclical) order. We say $p \in \poly V$ if $p$ is strictly inside the polygon.

The following Lemma~\ref{lem:combined} combines a number of simple but useful propositions.

\begin{lemma}\label{lem:combined}
    For a set of points $P$, the following propositions are true:
\begin{enumerate}
    \item \label{prop:active-in-three}
    Any point $p \in P$ 
    is active for at most three points on the first layer.
    
    \item \label{prop:no_active_points_in_active_triangle}
    Let $\Tri{t}{u}{v}$ be a triangle for consecutive vertices $(t,u,v)$ on the first layer and let $p \neq q$ be points $p \in \Tri{t}{u}{v}$ and $q \in \Tri{t}{p}{v}$. Then $q \notin A(u)$.

    \item \label{prop:monotonicity_of_layering}
    Let $p$ be a point on any layer $k$. After deleting any point $q \neq p$ and reconstructing the convex layers, $p$ is on layer $k-1$ or $k$.

    \item\label{prop:adjacent-triangles-sensitivity-change-only}
    Let $(t, u, v)$ be consecutive vertices on the first layer $L^1$. Then if $u$ is deleted, among the vertices in $L^1$, only the sensitivities of vertices $t$ and $v$ change.

    \item \label{prop:small-intersection}
    For adjacent points $(u, v)$ on the hull, $|A(u) \cap A(v)| \leq 1$.
\end{enumerate}

\end{lemma}
\begin{proof}
    See Section~\ref{sec:app_lemma_combined} in the appendix.
\end{proof}

\subsection{Bounding the active points}
We will show that once a point is active for a hull point, it remains active for that hull point until the point is moved to the first layer. This implies a much stronger result by Lemma~\ref{lem:combined}(1): over the entire course of the algorithm, a point becomes active for at most three other points. To do so, we first show that for each peel the active points $A(u)$ remain in $u$'s triangle (Lemma~\ref{lem:still-in-triangle}) and second that the points in $A(u)$ remain active (Lemma~\ref{lem:still-active}). 


\begin{lemma}\label{lem:still-in-triangle}
    Given a set of points $P$, for all adjacent hull points $(u, v)$ and for all points $p \in A(u) \setminus A(v)$, if $v$ is deleted then $p$ still remains within $u$'s triangle.
\end{lemma}
\begin{proof}
    Let $t$ be $u$'s other neighbor, and w.l.o.g. let the clockwise order on the hull be $(t, u, v)$. Then if $v'$ is $u$'s new neighbor after deleting $v$, the clockwise order on the new hull will be $(t, u, v')$. Because $p$ is active for $u$ before $v$ is deleted, $p \in \Tri tuv$.
    
    First, we consider the case where $v' \notin \Tri{t}{u}{v}$. We want to show that $p \in \Tri{t}{u}{v'}$. Equivalently, that $p$ is in the intersection of the three half-planes $\hp{tu}$, $\hp{uv'}$, and $\hp{tv'}$. Clearly, $p$ must satisfy the half-planes $\hp{tu}$ and $\hp{uv'}$ as these coincide with hull edges. In addition, since $v' \notin \Tri tuv$, the half-plane for $\hp{tv}$ is a subset of the half-plane for $\hp{tv'}$. Because $p \in \Tri tuv$, $p$ satisfies $\hp{tv}$. Therefore, $p$ must satisfy $\hp{tv'}$. 

    Now we consider the case where $v' \in \Tri tuv$. Assume that $p \notin \Tri tu{v'}$. Then because we know that $p \in \Tri tuv$, either $p \in \Tri t{v'}v$ or $p \in \Tri u{v'}v$. If $p \in \Tri t{v'}v$, by Lemma~\ref{lem:combined}(2), $p$ could not have been active for $u$ prior to deleting $v$. If $p \in \Tri u{v'}v$, $p$ is now outside of the convex hull. Either way, this is a contradiction.

\end{proof}

The following Lemma~\ref{lem:still-active} shows that if $p$ is in $A(u)$, it remains in $A(u)$ until moved to the first layer, after which it never becomes active again. 
It also shows that the active points $A(u)$ only change by adding or deleting points from either end, and thus can easily be found.
\begin{lemma}\label{lem:still-active}
    Given a set of points $P$, for all hull points $u$ and $v$ and for all points $p \in A(u) \setminus A(v)$, upon deleting $v$, $p$ is in $A(u)'$, $u$'s new set of active points.
\end{lemma}
\begin{proof}$ $\newline

    \begin{case}[$u$ is not adjacent to $v$]\label{case:u-and-v-not-adj}
    \end{case}
    If $u$ is not adjacent to $v$, there are no changes to $\TriS{u}$ upon deleting $v$, and thus, $A(u) = A(u)'$. 
    
    For the following cases, assume that $u$ was adjacent to $v$. Then by Lemma~\ref{lem:still-in-triangle}, $p$ is still in the triangle defined by $u$ even after deleting $v$. Also, w.l.o.g. let $(u, v)$ be the clockwise ordering of the points, and let $v'$ be $u$'s new neighbor.

    \begin{case}[$v' \in A(u)$]\label{case:v'-in-A(u)}  
    \end{case}
    By Lemma~\ref{lem:combined}(5), $A(u) \cap A(v) = v'$.
    By Lemma~\ref{lem:still-in-triangle}, all points $A(u) \setminus \lrCurlyBrackets{v'}$ are in $\Tri{t}{u}{v'}$. 
    Because the second layer is a convex hull, each consecutive pair of points $(a,b)$ in $t \circ A(u)$ define a half-plane $\hp{ab}$ with only points from the first layer to the left of each half-plane. 
    This is still the case after deleting $v$ by Lemma~\ref{prop:monotonicity_of_layering}. 
    Since the only new points on the first layer are $A(v)$ then all points in $A(u) \setminus \lrCurlyBrackets{v'}$ remain on the second layer. 
    Thus, the gift-wrapping starting from $t$ wraps around all points in $A(u) \setminus \lrCurlyBrackets{v'}$. 
    Gift wrapping can hit no new points because, if that were true, there must be some point on the second layer to the left of one of the half-planes in described above.
    Thus, $A(u)' = A(u) \setminus \lrCurlyBrackets{v'}$.

    \begin{case}[$v' \notin A(u)$]\label{case:v'-not-A(u)}    
    \end{case}
    Let $u_e$ be the last point $A(u)$.
    Similar to the previous case, the gift-wrapping certifies all points in $A(u)$. 
    Again, wrapping will not hit new active points before wrapping around $u_e$ because that would imply the points hit were to the left of the half-planes described previously. 
    When wrapping continues around $u_e$, several new active points may appear, until the wrapping terminates at $v'$. Thus, $A(u) \subseteq A(u)'$.
\end{proof}

\begin{theorem}\label{thm:bounding_active_points_amortized}
    For any 2D convex hull peeling process on $n$ points the total number of times any point becomes active in any triangle is at most $3n$.
\end{theorem}

\begin{proof}
   This follows directly from the results of Lemma~\ref{lem:combined}(1) and Lemma~\ref{lem:still-active}.
\end{proof}

\subsection{Updating sensitivities}
Next, we show that the total time to update the sensitivities in line~\ref{alg:updating_neighbour_areas} when peeling all $n$ points takes $\Oh{\Delta + n \logn}$ time. Here $\Delta$ is the the number of times any point becomes active for any triangle. Theorem~\ref{thm:bounding_active_points_amortized} proves that $\Delta = \Oh{n}$. 
The following lemma shows that the sensitivity of a point $u$ can be updated in time proportional to the increase to $\lrSize{A(u)}$ and an additive $\Oh{\logn}$ term. 
Figure~\ref{fig:compute_new_sensitivity} shows an example of how the sensitivity of a point changes when its neighbor is peeled.

\begin{lemma}\label{lem:update-sensitivities}
    Let $(u, v)$ be points on the first layer. Consider a peel of $v$ where $\delta_u$ new points become active points for $u$.
    Then the updated sensitivity $\sens{u}$ can be computed in $\OhTheta{\delta_u + \log n}$ time, excluding the time to restore the second and first layer. 
\end{lemma}
\begin{proof}
    The sensitivity $\sens{u}$ is equal to the area of the polygon $U = \poly{t \circ u \circ v \circ A(u)}$.
    By the \emph{shoelace formula}, the area of $U$ can be computed as the sum $S(U)$ of certain simple terms for each of its edges~\cite{2000_boland_Urrutia_polygon_chord_area_static, contreras1998cutting}. 
    We consider how $U$, and thus $S(U)$, changes when $v$ is peeled. 
    Inspecting the proof of Lemma~\ref{lem:still-active}, we see that at most two vertices are removed from $U$ and at most $1 + \delta_u$ vertices are added to $U$. 
    Furthermore, all the new vertices are located contiguously on the restored second layer and can be found in $\Oh{\delta_u + \logn}$ time using a tangent query from $u$'s new neighbor which replaces $v$.
    To update $\sens{u} = S(U)$, we simply add and subtract the appropriate $\Oh{\delta_u}$ terms depending on the removed and added edges.

\end{proof}

    \begin{figure}
        \centering
        \includegraphics[scale=0.16]{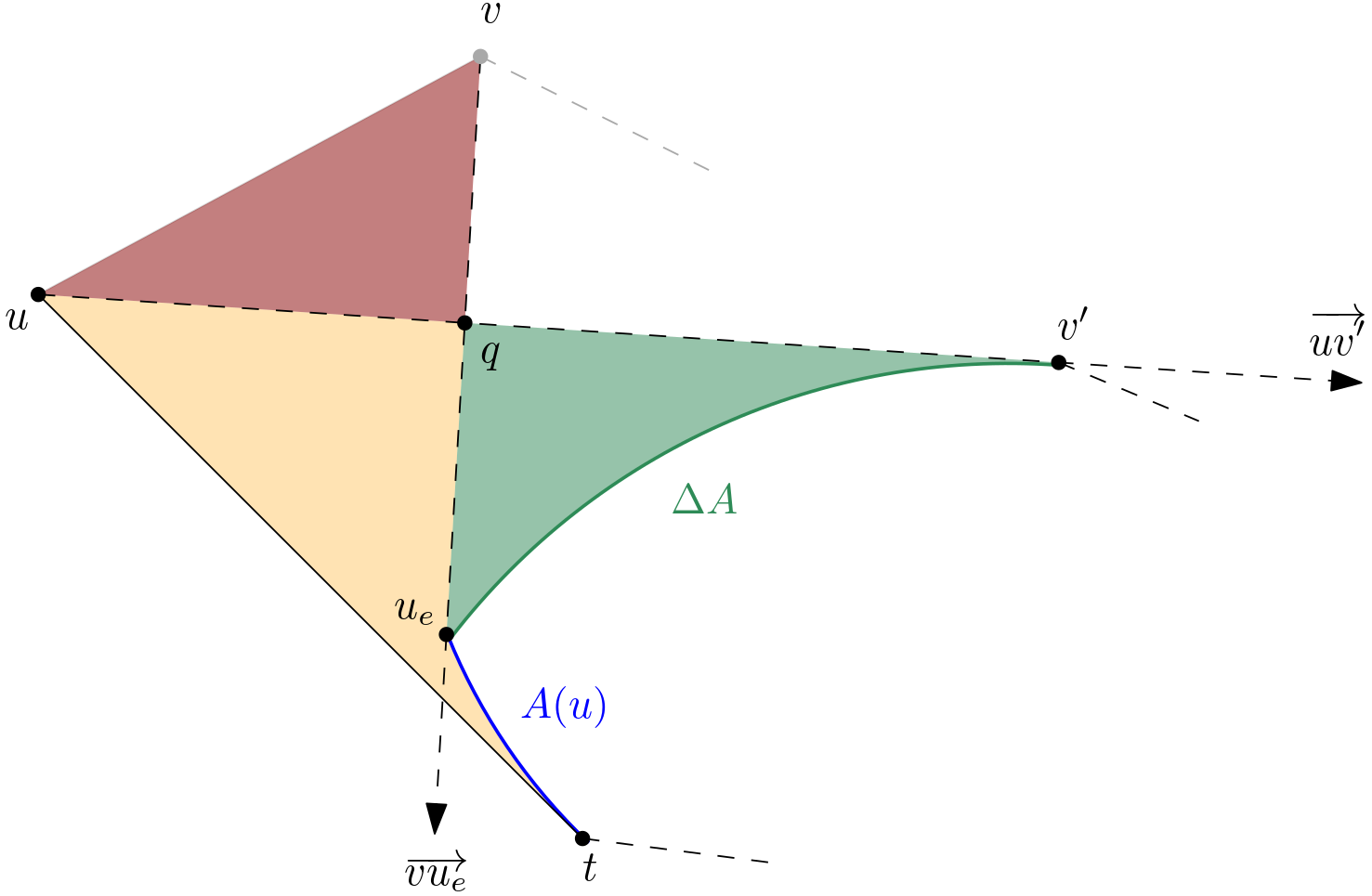}
        \caption{
        This figure shows how the sensitivity $\sens{u}$ changes when point $v$ is peeled. The point $q$ is the intersection of the tangent from $v$ to $u_e$ and the tangent from $u$ to $v'$, where $u_e$ is the last active point in $A(u)$ and $v'$ is the first active point in $A(v)$. After the peel, $v'$ replaces $v$ as $u$'s neighbor, and the points $\Delta A$ are newly active for $u$. The sensitivity $\sens{u}$ before peeling $v$ was equal to the area of $\poly{t \circ u \circ v \circ A(u)}$. After peeling $v$, the sensitivity $\sens{u}$ equals the area of $\poly{t \circ u \circ v' \circ \Delta A \circ A(u)}$. Note how this can be computed in $\Oh{\lrSize{\Delta A}}$ time from $\sens{u}$ before the peel of $v$ by subtracting the red area of $\Tri{u}{v}{q}$ and adding the green area of $\poly{u_e \circ q \circ v' \circ \Delta A}$.
        } \label{fig:compute_new_sensitivity}
    \end{figure}

\section{Generalization and open problems}
Theorem~\ref{thm:generalization} shows that Algorithm~\ref{alg:area_weighted_peeling} generalizes straightforwardly to other objectives such as peeling the point that causes the perimeter of the convex hull to decrease the most each iteration.

\begin{theorem}\label{thm:generalization}
    Let $u$ be a point and $O$ an objective where $\sensO{u}$ is the sensitivity of $u$ under $O$. Consider the following three conditions:
    \begin{enumerate}
        \item[\textbf{C1:}] If $u \notin L^1$, then $\sensO{u} = 0$.
        \item[\textbf{C2:}] If $u \in L^1$, then $\sensO{u} > 0$, and $\sensO{u}$ depends only on $u$, $u$'s neighbors and its active points $A(u)$.
        \item[\textbf{C3:}] 
        If a single point $p$ is added or removed from $A(u)$, then provided $\sensO{u}$ and the neighbors $a_i$ and $a_{j}$ of $p$ in $A(u)$, the new sensitivity $\sensO{u}'$ can be computed in $\Oh{\logn}$ time.
    \end{enumerate}
    If $O$ satisfies the above conditions, then Algorithm~\ref{alg:area_weighted_peeling} runs in $\Oh{n\logn}$ time for objective $O$. 
\end{theorem}

\begin{proof}
By conditions $\textbf{C1}$ and $\textbf{C2}$, it is always a point $u$ on the first layer that is peeled. 
Furthermore, when $u$ is peeled only the sensitivities of the new points on the first layer and the neighbors of $u$ must be updated since they are the only points for which their active points or neighbors change. Thus, Algorithm~\ref{alg:area_weighted_peeling} can be used for objective $O$. Now we will show that the runtime of Algorithm~\ref{alg:area_weighted_peeling} remains $\Oh{n\logn}$.

First, observe that all parts unrelated to computing sensitivities behave the same and still take $\Oh{n\logn}$ time. By condition $\textbf{C3}$, for a point $u$ on the first layer, its sensitivity $\sensO{u}$ only depends on its neighbors and active points $A(u)$. As described in the proof of Lemma~\ref{lem:update-sensitivities}, when the set of points that affect $\sensO{u}$ changes, these points are readily available. The total number of neighbor changes is $\Oh{n}$ since, in each iteration, only the neighbors of the points adjacent to the peeled point change. 
The total number of changes to active points is $\Oh{n}$ by Theorem~\ref{thm:bounding_active_points_amortized}. If there are multiple changes to the active points in one iteration, such as when deleting one of $u$'s neighbors, we perform one change at a time and, by condition $\textbf{C3}$, the total time to update sensitivities is $\Oh{n\logn}$. 
\end{proof}

For concrete examples, we show how the three objectives \emph{area} ($O_A$), \emph{perimeter} ($O_P$), and \emph{number of active points} ($O_N)$ fit into this framework. 

Let $\f{\sens{u}, a_i, p, a_{j}} = \sens{u} - \dist{a_i}{a_{j}} + \dist{a_i}{p} + \dist{p}{a_{j}}$ be a function for computing the sensitivity $\sens{u}$ when $p$ is added to $A(u)$ between $a_i$ and $a_{j}$ (the functions where a point is removed from $A(u)$ or a neighbor of $u$ changes are similar).  
For $f$ to match each of the objectives it is sufficient to implement $\dist{\cdot}{\cdot}$ as follows for points $a,b\in\R^2$:

\begin{enumerate}
    \item[$O_A$:] $\dist{a}{b} = \frac{1}{2} \lrParents{a_2 b_1 - a_1 b_2} $
    \item[$O_P$:] $ \dist{a}{b} = \sqrt{\lrParents{b_2 - a_2}^2 + \lrParents{b_1 - a_1}^2} $
    \item[$O_N$:] $ \dist{a}{b} = 1 $
\end{enumerate}

The case with $O_A$ is based on the shoelace formula. Additionally, for $O_N$ to satisfy condition \textbf{C2}, we add $1$ when computing the sensitivity of $u \in L^1$ to ensure that $\sens u > 0$ even if $|A(u)| = 0$.
For the three objectives, $f$ takes $\Oh{1}$ time to compute satisfying the $\Oh{\logn}$ time requirement from condition $\textbf{C3}$.

\subsection{Open problems}
The first open problem is extending the result to $\R^3$ or higher. Directly applying our approach requires a dynamic 3D convex hull data structure, and Theorem~\ref{thm:bounding_active_points_amortized} has to be extended to 3D. Second, is it possible to improve the quality of peeling by performing $z$-peels, even for $z = 2$ in $\OhSmall{n}$ time? Third, is there an efficient approximation algorithm for $k$-peeling? 

\section*{Acknowledgement}
We thank Asger Svenning for the initial discussions that inspired us to consider this problem.

\small
\bibliographystyle{abbrv}
\bibliography{refs}




\section*{Appendix}

\subsection{Proof of Lemma~\ref{lem:combined}}\label{sec:app_lemma_combined}
\begin{customlemma}{\ref{lem:combined}(1)}
    Fix a point set $P$. Any point $p \in P$ 
is active in at most three triangles.
\end{customlemma}\begin{proof}
    First, note that a point can only be active for a hull point $u$ if it is located inside $\TriS u$, so it is sufficient to show that any $p$ is strictly inside at most three triangles. In addition, one can prove this by showing that $\TriS u$ only intersects with its neighbors' triangles $\TriS t$ and $\TriS v$. 

    Consider some $\TriS z$, such that $z$ is not a neighbor of $u$. That is, $u$ is not one of the vertices of $\TriS z$. If $\TriS z$ intersects with $\TriS u$, then either a vertex of $\TriS z$ is inside $\TriS u$ or the convex hull is a self-intersecting polygon, both violating convexity.
\end{proof}

\begin{customlemma}{\ref{lem:combined}(2)}
    Let $\Tri{t}{u}{v}$ be a triangle for consecutive vertices $(t,u,v)$ on the first layer and let $p \neq q$ be points $p \in \Tri{t}{u}{v}$ and $q \in \Tri{t}{p}{v}$. Then $q \notin A(u)$.
\end{customlemma}
\begin{proof}
    By definition, $p \in \poly{t \circ v \circ A(u)}$ or $p \in A(u)$. 
    Either way, $q \in \Tri{t}{p}{v}$ implies that $q \in \poly{t \circ v \circ A(u)}$, so $q \not \in A(u)$. 
\end{proof}

\begin{customlemma}{\ref{lem:combined}(3)}
    Let $p$ be a point on any convex layer $k$. After deleting any point $q \neq p$ and reconstructing the convex layers, $p$ is on layer $k-1$ or $k$.
\end{customlemma}
\begin{proof}
    First we show that $p$ never moves inward to layer $k' > k$. Consider the outermost layer $L^1$. By a property of convex hulls, every point $v$ inside the convex hull is a convex combination of the hull points whereas any point $u \in L^1$ is not a convex combination of $L^1 - \{u\}$. If deleting $q$ causes $p \in L^1$ to descend to a layer inside $L^1$, that implies that $p$ is a convex combination of some subset of $P - \{q, p\}$. This contradicts the fact that $p$ is not a convex combination of $L^1 - \{p\}$ and by extension is not a convex combination of $P - \{p\}$. Because of the recursive definition of convex layers, the proof for subsequent layers is symmetric. 

    Now we will show that $p$ never moves up more than one layer at a time. This is clearly true for $L^1$ and $L^2$ because only one point is completely removed from the structure at at time (i.e. shifts to layer 0). For layers $k \geq 3$, consider a point $p$ on layer $k$ that moves to layer $k' \leq k-2$. Let $L^{*k'}$ be the set of points on layer $k'$ after deleting $q$. Let $L^{k-1}$ be the set of points on layer $k-1$ before deleting $q$.
    
    Because $p \in L^{*k'}$, no convex combination of the points in $L^{*k'} - \{p\}$ equals $p$ by convexity. By the inductive hypothesis, all points on $L^{k-1}$ are convex combinations of $L^{*k'}$ because upon deleting $q$ no point on $L^{k-1}$ advances above layer $k'$. Furthermore, they are all convex combinations of $L^{*k'} - \{p\}$ as $p$ itself is a convex combination of $L^{k-1}$. But if $p$ is not a convex combination of $L^{*k'} - \{p\}$, and all the points on layer $k-1$ are convex combinations of $L^{*k'} - \{p\}$, then prior to deleting $q$, $p$ was above layer $k-1$, which is a contradiction. 
\end{proof}

\begin{customlemma}{\ref{lem:combined}(4)}
    Let $(t, u, v)$ be consecutive vertices on the first layer $L^1$. Then if $u$ is deleted, among the vertices in $L^1$, only the sensitivities of vertices $t$ and $v$ change.
\end{customlemma}
\begin{proof}
    Consider a vertex $z$ not adjacent to $u$. By the same arguments as in the proof of Lemma~\ref{lem:combined}(1), the vertices defining $\TriS z$ do not change upon deleting $u$ because it does not intersect $\TriS u$. In addition, because their triangles do not intersect, $|A(u) \cap A(z)| = 0$. Therefore, no points are removed from $A(z)$ upon deleting $u$. 

    Lastly, we will show that no points are added to $A(z)$ upon deleting $u$. Assume that there is some point $p$ added to $A(z)$ when we delete $u$. But if $p$ satisfies the conditions of being active for $z$ and $\TriS z$ did not change upon deleting $u$, it should have been active for $z$ before $u$ was deleted, which is a contradiction.

    Because $\TriS z$ and $A(z)$ do not change upon deleting $u$, it must be that $\sens z$ remains the same.
\end{proof}


\begin{customlemma}{\ref{lem:combined}(5)}
    For adjacent points $(u, v)$ on the hull, $|A(u) \cap A(v)| \leq 1$.
\end{customlemma}
\begin{proof}
We assume the contrary. Let $p \neq p'$ be two points such that $p, p' \in A(u) \cap A(v)$. By the definition of \emph{active} and Lemma~\ref{lem:combined}(3), $p$ and $p'$ must be on the second layer. W.l.o.g. let $(u, v)$ be the clockwise ordering of the points on the first layer. In addition, let $t$ be $u$'s counterclockwise neighbor. 

Say that $p$ is the first point in $A(v)$. Then we have the tangent line $\hp{up}$ that defines $p$. By definition of tangent lines, no point on the second layer can be to the left of $\hp{up}$. But for $p'$ to be active for $v$, then $p'$ must be to the left of $\hp{pv}$. The only way to satisfy both half-planes is for $p'$ to be placed such that $p \in \Tri{t}{p'}{v}$, in which case by Lemma~\ref{lem:combined}(2) $p$ cannot be in $A(u)$, which is a contradiction. 

\end{proof}
\end{document}